\documentclass[prl,floatfix,amsmath,superscriptaddress,twocolumn]{revtex4}
\usepackage{amssymb}
\usepackage{graphicx}
\usepackage{graphics}
\usepackage{amsmath}
\usepackage{amsthm}
\usepackage{color}

\def\sz{\sigma_3}

\newcommand{\bea}{\begin{eqnarray}}
\newcommand{\eea}{\end{eqnarray}}
\def\bi{\begin{itemize}}
\def\ei{\end{itemize}}
\def\bc{\begin{center}}
\def\ec{\end{center}}

\def\C{\hbox{$\mit I$\kern-.7em$\mit C$}}
\def\R{\hbox{$\mit I$\kern-.6em$\mit R$}}

\def\ket#1{|#1\rangle}
\newcommand{\one}{\mbox{$1 \hspace{-1.0mm}  {\bf l}$}}

\def\ket#1{\left| #1\right>}
\def\bra#1{\left< #1\right|}

\newtheorem{theorem}{Theorem}

\newtheorem{lemma}[theorem]{Lemma}

\begin{document}
\author{C. Spee}
\affiliation{Institute for Theoretical Physics, University of
Innsbruck, Innsbruck, Austria}
\author{J. I. de Vicente}
\affiliation{Departamento de
Matem\'aticas, Universidad Carlos III de Madrid, Legan\'es (Madrid),
Spain}
\author{B. Kraus}
\affiliation{Institute for Theoretical Physics, University of
Innsbruck, Innsbruck, Austria}
\title{Remote entanglement preparation}
\begin{abstract}

We introduce a new multipartite communication scheme, with the aim
to enable the senders to remotely and obliviously provide the
receivers with an arbitrary amount of multipartite entanglement. The
scheme is similar to Remote State Preparation (RSP). However, we
show that even though the receivers are restricted to local unitary
operations, the required resources for remote entanglement
preparation are less than for RSP. In order to do so we introduce a
novel canonical form of arbitrary multipartite pure states
describing an arbitrary number of qubits. Moreover, we show that if
the receivers are enabled to perform arbitrary local operations and
classical communication, the required resources can drastically be
reduced. We employ this protocol to derive robust entanglement
purification protocols for arbitrary pure states and show that it
can also be used for sending classical information.
\end{abstract}
\maketitle

Quantum information theory offers revolutionary ways to process and
transmit information. In order to understand the fundamental laws of quantum information processing one needs to investigate both, the classical as well as the  quantum resources which are required to achieve
certain tasks, such as quantum
communication \cite{LoPo99}. Naturally,
entanglement is an important resource in this context. Due to that
and its foundational interest an enormous effort has been made to
achieve a better understanding of entanglement \cite{entreview}. However, despite
considerable progress, many questions remain unanswered,
particularly in the realm of multipartite states \cite{entreview}.
In this context, the following questions, which we address in this
article, arise very naturally: Do there exist quantum communication
protocols which allow some parties to supply other spatially
separated parties with arbitrary multipartite entanglement? If the answer is yes, what are the required
resources and how do they compare to the ones required in other
communication protocols? Our purpose is two-fold, besides shedding
light in quantum information protocols and their cost, we aim at
obtaining a better understanding of multipartite entanglement and
its applications.

Questions as those posed above are the focus when addressing the
limits of quantum information processing in fundamental protocols
such as teleportation \cite{tele}, dense coding \cite{Be93}, and Remote State Preparation (RSP) \cite{Lo00,BennettDiVincenzo01}. As pointed out in \cite{LoPo99} it is necessary to consider both, the classical as well as the quantum resources in this context and it would be misleading to ignore one of them. The aim of the bipartite communication protocol, RSP is to enable the sender, Alice ($A$), who is spatially
separated from the receiver, Bob ($B$), to prepare B's system in a
certain state {\it known to $A$} using entanglement and classical
communication \cite{Lo00,BennettDiVincenzo01}. The protocol is
called faithful if $B$ obtains deterministically the desired state
and it is called oblivious if there is no leaking of further
information about the state to $B$ than the information that is
already contained in the state itself. One distinguishes between the asymptotic case \cite{BennettDiVincenzo01,BennettHayden04} and 
the single copy case. In this article we will focus on the latter, where it has been proven that the resources needed to perform oblivious and faithful RSP for a generic ensemble of states coincide with the resources required for quantum
teleportation \cite{tele}, i.e. $1$ ebit and $2$ cbits per
transmitted qubit \cite{LeSh03,Lo00}. Thus, the fact that $A$ knows
the state to be prepared does not help in reducing the required
resources, unless one would restrict the set of states to be
transmitted \cite{Lo00,Pati00}. 

In this paper we introduce a multipartite communication scheme. In contrast to RSP, the aim here is not to enable $A$ to prepare an arbitrary state for $B$, but to provide $B$, who is decomposed into several spatially separated $B_i$, with an arbitrary amount of multipartite entanglement. All parties, the ones in $A$ ($A_i$) and, more importantly, the ones in $B$ ($B_i$) are restricted to local operations, which makes the entanglement received by $B$ a useful resource. The classical communication within $A$ and $B$ resp. is considered to be for free. We will show that this multipartite Remote Entanglement Preparation (REP) requires less resources than RSP, despite the restriction to local operations.  As we will see, the REP protocol can be easily turned into a RSP protocol. However, in order to stress the difference between REP and RSP let us note that in the case where $B$ might learn about the state, RSP can be performed using solely classical communication. Contrary, REP cannot be performed since the various $B_i$ are spatially separated and are therefore restricted to Local Operations assisted by Classical Communication (LOCC), which does not allow them to create any entangled state.

The outline of the paper is the following. First of all, we present the main ideas of REP by considering the simplest example of bipartite entanglement preparation. In order to generalize it to multipartite systems we introduce a novel Canonical Form (CF) for arbitrary $n$--qubit states, which depends on $P_n$ (see below) parameters. Any $n$--qubit state can be transformed into its CF via Local Unitary operators (LUs). Then, we show that any $n$--qubit state in the CF can be generated deterministically by measuring locally $P_n$ qubits of a certain $(n+P_n)$--qubit resource state. These properties will ensure that if $A$ and $B$ would share this resource state, where $A$ holds the $P_n$ qubits to be measured, she is able to prepare B's system in an arbitrary state in the CF up to LUs, which are performed by $B$ after receiving the classical information about those LUs of $A$. We show that this protocol has the following properties: (I) it is oblivious, (II) it is faithful (III) all actions performed by $A$ and $B$ are local (IV) the number of
ebits shared by $A$ and $B$ is $n$ (V) the number of cbits needed
to be communicated from $A$ to $B$ is $2n-1$, which is even less than what is required for RSP  \cite{LeSh03}. The parties constituting $B$, knowing that they will receive a state in the CF can then, in case this is required apply LUs to obtain an arbitrary state. However, in most cases, this will, arguably not be required, since $B$ will make use of the entanglement contained in the state. That is, the resource is entanglement, not a particular state. Furthermore, we go one step beyond that and provide a protocol for remote maximally entangled state preparation. In this context we use the notion of maximally entangled sets \cite{dVSp13}. Given a maximally entangled set for $n$--partite states, any other $n$--partite state can be prepared deterministically via LOCC.
We show that in the case of $3$--qubit states, this allows to reduce
 the required resources drastically. Moreover, we emphasize the relevance of the new truly multipartite communication scheme by demonstrating that it can be utilized to derive robust purification protocols for arbitrary states and to transmit classical information.

Let us explain the main idea of the protocol by considering the first non--trivial example of preparing arbitrary bipartite entanglement. Due to the existence of the Schmidt decomposition \cite{Nielsen-Chuang} an arbitrary bipartite state can be written (up to LUs) as
$\ket{\Psi(\alpha)}=Z_{12}(\alpha)\ket{+}^{\otimes 2}$ where $\alpha \in \R$ and $\ket{\pm}$ denotes the eigenbasis of $\sigma_1$. Here, and in the following $Z_{i_1,\ldots ,i_k}(\alpha)$ denotes the phase gate $e^{i\alpha \sz^{(i_1)}\otimes \ldots \otimes \sz^{(i_k)}}$ acting non--trivially on the $k$ systems $i_1,\ldots, i_k$ and the Pauli operators are denoted by $\sigma_i$, with $\sigma_0=\one, \sigma_1=\sigma_x, \sigma_2=\sigma_y, \sigma_3=\sigma_z$. REP can be achieved as follows. Initially $A$ and $B$ share the
resource state $H^{\otimes 3}\ket{GHZ}_{A,B_1,B_2}$, where $H$ denotes the Hadamard gate and $\ket{GHZ}\propto \ket{000}+\ket{111}$. $A$ chooses the value of $\alpha$, depending on how much entanglement she wants to prepare in $B's$ system. She performs a Von Neumann measurement on her
qubit, $A$, in the basis

\begin{eqnarray}\label{eqbasis}\mathcal{B}_{\alpha}=\{\ket{\phi^i(\alpha)}\equiv \sigma_3^i Z(-\alpha)\ket{+}\}_{i=0,1},\end{eqnarray} where $i$ denotes the measurement outcome. It is easy to verify that if $A$'s measurment outcome is $0$ ($1$), $B's$ system is prepared in the state
$\ket{\Psi(\alpha)}$ ($\sz\otimes\sz\ket{\Psi(\alpha)}$) resp.. Hence, once $A$ tells $B$ the measurement outcome, $B$ can apply the Pauli
operators locally to obtain the state $\ket{\Psi(\alpha)}$ deterministically. The
protocol is oblivious since both, the transmitted classical information and the resource state are
independent of the prepared state. Since the classical information
within $B$ is for free, the required resources are $\frac{1}{2}$
ebits and $\frac{1}{2}$ cbits per transmitted qubit. In contrast to
that, RSP requires $1$ ebit and $2$ cbits per qubit \cite{LeSh03}.

We generalize this protocol now to arbitrary many qubits. One of the difficulties here is that there does not exist a simple generalization of the Schmidt decomposition to $n$--partite systems, which reflects the difference between bipartite and multipartite states. Thus, we first have to introduce a Canonical Form (CF) for arbitrary $n$--qubit states such that each state is LU--equivalent to its CF. As we will show, for $n\geq 3$, the CF depends on $P_n=2 ^{n+1}+ 2^{n-3}-3(n+1)$ real parameters (phases), which we denote by $\alpha_i$. The corresponding CF will be denoted by $\ket{\Psi_n(\{\alpha_i\})}$ or simply by $\ket{\Psi_n}$. Then, we derive a ($P_n+n$)--qubit resource state in analogy to the GHZ--state in the example above. This resource state, $\ket{\Phi_n}$, has the property that performing local Von Neumann measurements on $P_n$ of the qubits in the bases $\mathcal{B}_{\beta_i}$, given in Eq. (\ref{eqbasis}) for properly chosen phases, $\beta_i=\pm \alpha_i$ \footnote{It will become clear later how the signs have to be chosen.} leads, for any measurement outcome, to a state \bea \label{Bstates} \sigma_{i_1}\otimes \sigma_{\bf i}\ket{\Psi_n(\{\alpha_i\})},\eea where $i_1\in \{0,3\}$ and ${\bf i}\in \{0,1,2,3\}^{n-1}$ depend on the measurement outcomes. Here, and in the following we use the notation $\sigma_{\bf j}=\sigma_{j_1}\otimes \ldots \sigma_{j_k}$, where ${\bf j}$ is a $k$--dimensional vector (with $k$ being specified unless it is clear from the context) with entries $j_l\in \{0,1,2,3\}$. With all that, the REP protocol works then as follows:
(i) Initially $A$ and $B$ share the resource state $\ket{\Phi_n}$, where $A$ is holding the $P_n$ qubits which need to be measured and $B$ the other $n$ qubits.
(ii) $A$ chooses the phases $\{\alpha_i\}$ and therefore the entanglement $B$ receives and performs local measurements on her qubits in the bases $\mathcal{B}_{\beta_i}$, given in Eq. (\ref{eqbasis}) for $\beta_i=\pm \alpha_i$. This prepares $B$'s system in one of the states in Eq. (\ref{Bstates}).
(iii) $A$ sends the classical information $i_1\in \{0,3\}$ and ${\bf i}\in \{0,1,2,3\}^{n-1}$ to $B$.
(iv) $B$ applies the local Pauli operator $\sigma_{i_1}\otimes \sigma_{\bf i}$ to obtain the desired state.

Note that for any measurement outcome, only two out of the four Pauli operators occur on the first qubit [see Eq. (\ref{Bstates})]. This is the reason, why the required classical communication is only $2n-1$ cbits. Since neither this information nor the resource state depend on the state $A$ wants to prepare (like the GHZ--state did not in the example above) this faithful protocol is also oblivious. Moreover, all actions performed by $A$ and $B$ are local and the number of
ebits shared by $A$ and $B$ is $n$ \footnote{This can be seen using the fact that the resource state is a stabilizer state, as we will show later on.}. Note that $B$ has to
know that he receives a state in the CF since otherwise his state
would be completely mixed due to the randomness of the LUs. However, as we will show, the
description of the CF is a side information that is independent of
$n$ and it can be agreed upon by $A$ and $B$ beforehand as part of
the protocol. Thus, even though $B$ receives an arbitrary amount of
multipartite entanglement, this multipartite protocol outperforms RSP despite the fact that all actions are local
single qubit operations.

Let us now introduce a CF of $n$--qubit states. In \cite{dViCa12} we have shown that an arbirary $3$--qubit state is (up to LUs) of the form
\begin{eqnarray}\label{standform3}\ket{\Psi_3}=Z_{13}(\alpha_1) Z_{12}(\alpha_2) (T_2(\alpha_3,\alpha_4) \otimes T_3)Z_{23}(\alpha_5)\ket+^{\otimes 3},\end{eqnarray} where $T_3=  e^{-i\frac{\pi}{4}\sigma_1 }Z(-\frac{\pi}{4})H$, $T_2(\alpha_3,\alpha_4)=e^{i\frac{\pi}{4}\sigma_1 }Z( \alpha_3) e^{-i\frac{\pi}{4}\sigma_1} Z(\alpha_4)H $. Here, and in the following, the qubits are
numerated consecutively and the subscript of any operator different than the Pauli operators indicates on which system it is acting on. We generalize now this CF for an arbitrary number of qubits. The aim is to achieve a systematic generalization such that, as in Eq. (\ref{standform3}), the only gates occurring in the CF are local Clifford gates \cite{Nielsen-Chuang} and phase gates and that on the first qubit only phase gates are acting on. As we will see later, this will ensure that the states $B$ obtains are of the form given in Eq. (\ref{Bstates}).

We consider first the $4$ qubit case. W. l. o. g we write
an arbitrary $4$--qubit state (up to a LUs) as $ \sqrt{2}^{-1}(\ket{0}\ket{\chi_0}+\ket{1}U_2 U_3 U_4\ket{\chi_1}),$ where both, $\ket{\chi_{0,1}}$ are normalized $3$--qubit states in the CF [Eq. (\ref{standform3})] and for $i=2,3,4$, $U_i=V_i ^\dagger Z(2\alpha_i)V_i$, with $V_i=e^{i\beta_i \sigma_2}Z(\gamma_i)$ for some phases $\beta_i,\gamma_i$. Using the CF of the three qubit states and the fact that a controlled 3--qubit phase gate can be written as a product of a 3--qubit and a 2--qubit phase gates, 
it is straightforward to find the following CF for $4$--qubit states,

\bea\nonumber \ket{\Psi_4} = \prod_{i=2}^4 Z_{1i}(\alpha_{i})\prod_j W^j \ket+^{\otimes 4},\nonumber \eea
where $W^j$ are either local Clifford gates acting on qubits $2,3,4$ or phase gates acting on up to three qubits (including qubit 1).

In the same way the CF of an arbitrary number of qubits, $n$, can be derived. As before, it is obtained by the CF of $(n-1)$--qubit states and is therefore defined via the CF of $3$--qubit states [Eq. (\ref{standform3})]. Hence, the CF can be easily described and constructed for arbitrary $n$. The number of required parameters ($\alpha_i$), $P_n$, is given recursively by $2P_{n-1}+3(n-1)$, with $P_3=5$, which leads to $P_n=2^{n+1}-3(n+1)+2^{n-3}$.
Due to the construction only local Clifford and phase gates occur in the decomposition and on the first qubit only phase gates are acting on.

In order to derive now the resource state, $\ket{\Phi_n}$ we use the notion of deterministic gate implementation (see also
\cite{DuBr07,OnewayQC,AndersBrowne09}). In \cite{DuBr07} it has been
shown how arbitrary phase gates can be implemented deterministically
on an arbitrary unknown input state, $\ket{\xi}$. We briefly recall
the idea here and then generalize it to arbitrary gates. To implement a single qubit phase gate, $Z(\alpha)$, deterministically, one might use the 3--qubit GHZ-state as an auxiliary state and perform a measurement in the Bell basis $\{\sigma_i\otimes \one \ket{\Phi^+}\}_{i=0}^3$, with $\ket{\Phi^+}\propto \ket{00}+\ket{11}$ on one of the qubits, say qubit $3$ and the input qubit. Depending on the measurement outcome, $i\in \{0,1,2,3\}$, party $2$ chooses $\beta$ to be $+\alpha$ for $i=0,3$  and $-\alpha$ for $i=1,2$ and measures in the basis $\mathcal{B}_{\beta}$ given in Eq. (\ref{eqbasis}). It is easy to see that the resulting state is $\sigma_3^k\sigma_i Z(\alpha)\ket{\xi}$, where $k \in\{0,1\}$ denotes the measurement outcome of party $2$. Thus, choosing the  measurement basis of party $2$ depending on the outcome of the Bell-measurement allows to implement the gate deterministically, since $\sigma_3^k \sigma_i$ can be applied to obtain the desired state. In the following we will call this operator correction operator and the qubit, which has to be measured in the basis ${\cal B}_{\beta }$ in order to implement the gate, controlling qubit.

This idea can be easily generalized to implement phase gates acting on arbitrary many qubits \cite{DuBr07}. Similarly to above, one uses the fact that for any $m$--qubit phase gate, $Z_m(\alpha)$, we have
$Z_m(\alpha) \sigma_{\bf j}=\sigma_{\bf j} Z_m((-1)^l \alpha)$, with $l\in\{0,1\}$ for any ${\bf j}\in \{0,1,2,3\}^m$. Again only one controlling qubit is required which is measured in the basis $\mathcal{B}_{(-1)^l \alpha}$. If $k$ denotes the measurement outcome, then similarly to above, the local correction operators are $(\sigma_3^{\otimes m})^k \sigma_{\bf i}$, where ${\bf i}$ denotes the outcome of the Bell measurement.
\begin{figure}[h!]
\begin{center}
  \includegraphics[scale=0.30]{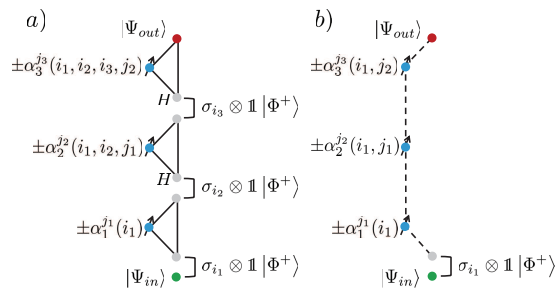}
  \caption{Deterministic gate implementation of a single qubit gate: In a) the 3 gates are implemented consecutively. However, the Bell--measurements performed here can be done independently of the input state, which leads to the $5$--qubit state presented in b). Note that the local Clifford gates (here $H$) only alter the Pauli operators. $i_k\in \{0,1,2,3\}$ denotes the measurement outcome of the Bell measurement and $j_k\in \{0,1\}$ the measurement outcome of the measurement on the controlling qubits. The signs of the measurement angles depend on the increments, $i_k, j_l$.}
\label{resource3}
\end{center}
\end{figure}

The gate implementation explained above can be generalized to implement arbitrary gates. In order to see that, let us first consider an arbitrary single qubit unitary, $U(\alpha_1,\alpha_2,\alpha_3)\equiv Z(\alpha_1) H Z(\alpha_2) H Z(\alpha_3)$. The gates could be implemented consecutively using the corresponding auxiliary states [see Fig. 1]. However, as illustrated in Fig. 1 the single qubit gate can also be implemented deterministically via a $5$--qubit state. Note that the order of the measurements has to be respected, since the signs of the measurement angles depend on the measurement outcome of the previous measurements (see Fig. 1). Note further that the $5$--qubit state (Fig. 1) is a stabilizer state \cite{stab}, since the GHZ state is a stabilizer state and all the measurements and the local Clifford gates map a stabilizer state to a stabilizer state.

Since any 2--qubit gate can be decomposed into local unitaries and a non--local part, which is of the form $e^{i\sum_{i=1}^3 \alpha_i \sigma_i\otimes \sigma_i}$ \cite{KrausCi01} such a gate can be easily decomposed into local Clifford and phase gates. Hence, the state required to implement this gate can be determined like in the case of single qubit gates. As any gate can be decomposed into single--and two--qubit gates, the successive implementation of gates explained here, allows to deterministically implement an arbitrary $n$--qubit gate. The corresponding auxiliary states are always $(p_n+2n)$--qubit stabilizer states, where $p_n$ denotes the number of phase gates (which are not Clifford gates) in the decomposition. Note that the value of the phases can be fixed at the very end by choosing the corresponding measurement basis.

For REP we construct now the resource state, $\ket{\Phi_n}$. We denote by $U_n$ the unitary such that
$\ket{\Psi_n}\equiv U_n \ket{+}^{\otimes n}$ is the CF of an $n$--qubit state. As we have shown before, $U_n$
can be decomposed into local Clifford gates and $P_n$ phase gates. Thus, a $(P_n+2n)$--qubit state is required to implement $U_n$ deterministically for an arbitrary input state. Hence, using $\ket{+}^{\otimes n}$ as an input state, leads to a
$(n+P_n)$--qubit resource state,
$\ket{\Phi_n}$, which can now be employed to generate any state $\ket{\Psi_n(\{\alpha_i\})}$. As shown above, measuring locally the $P_n$
controlling qubits in the bases ${\cal B}_{\pm \alpha_i}$ [see Eq. (\ref{eqbasis})] leads to one of the states in Eq. (\ref{Bstates}). The reason why only two out of four Pauli operators occur on the first qubit is that on the first qubit only phase gates (and no local Clifford gates) are acting on. 

Combining the CF and the construction of the corresponding resource state, we have now derived the REP protocol presented at the beginning.
Considering for instance the $3$--qubit case, it is straightforward to show that the resource state is
given by the $8$ qubit state (see also Fig. 2 in Appendix B),

\begin{align}
\label{res3}
\ket{\Phi_3}&=H_{6}Z_{8}(\frac{\pi}{4})H_{8}Z_{7}(-\frac{\pi}{4})Z_{2}(\frac{\pi}{4})S_{12}S_{23}\\ \nonumber &S_{37}S_{24}S_{25}S_{27}S_{46}S_{48}
S_{58}S_{67}S_{78}\ket+^{\otimes 8},
\end{align}
where $S_{ij} =\ket{0}\bra{0}\otimes\one+\ket{1}\bra{1}\otimes \sigma_3.$ Qubits $6$, $7$ and $8$ are held by B. Whereas, the first five qubits are the controlling qubits, which are held by $A$. She chooses freely the parameters $\alpha_i$ in Eq.
(\ref{standform3}) and therefore also the multipartite
entanglement $B$ obtains. First, qubit 1 is measured in
$\mathcal{B}_{\alpha_5}$ (see Eq. (\ref{eqbasis})), then qubit 2 in
$\mathcal{B}_{\pm \alpha_4}$ etc.,  As before, the measurements have to be performed in a certain order
to make sure that the protocol is deterministic. $A$ computes the
local correction operators and sends this information to $B$. This requires only
$\frac{5}{3}$ cbits per qubit, i.e. in total $1$ cbit less than in
RSP. The required
entanglement shared between A and B corresponds to 1 ebit per qubit \footnote{
    Note that the REP can be easily turned into a RSP, if the resource state is adapted to take LUs into account. Since then all Pauli operators can occur, 2 $n$ cbits are needed, as was shown to be required in \cite{LeSh03}.}.

So far we have considered the situation where $A$ wants to provide
$B$ with an arbitrary amount of multipartite entanglement. However, one might even go beyond that, as our only
requirement is that all the actions are local. Suppose that we allow $B$ to perform an arbitrary
LOCC protocol. In this case it is sufficient to provide $B$ with
those states, which are required to obtain any other state via LOCC.
For instance, if $A$
provides $B$ with the maximally entangled state, $\ket{\Phi^+}$, $B$
can afterwards transform this state deterministically into any other
bipartite state \cite{nielsen}. Since this can be achieved by sharing
initially a GHZ--state, the required resources are $\frac{1}{2}$
ebits and $\frac{1}{2}$ cbits per qubit, which coincides in this
case with the resources required for REP. Let us now show that the
resources can be drastically reduced in case of $3$--partite
entanglement. In \cite{dVSp13} we show that via LOCC any 3--qubit state can be obtained
from a state given in Eq.
(\ref{standform3}) with $\alpha_3=\alpha_4=\pi/4$. The
corresponding $6$--qubit resource state can be easily determined (see Fig. $3\, a$ in Appendix B). If $A$ and $B$ share this state, $A$ can remotely
and obliviously prepare any state with $\alpha_3=\alpha_4=\pi/4$ for $B$ by choosing appropriate
measurements on the 3 controlling qubits. In this case, $A$ just needs to
send 3 cbits to $B$ as she only needs to specify the outcomes of
3 local measurements. Moreover, as can be easily seen from Fig. 3 (b) in Appendix B, A and B share only 2 ebits. Thus, only one cbit and 2/3 ebits per qubit are required.

Finally, let us present some applications of the REP protocol introduced here, which emphasize the relevance of REP. For instance, it can be used to send classical information, which is encoded in the correction operators (see Appendix A). Moreover, REP can also be used to derive
robust entanglement purification protocols for arbitrary states (see Appendix B for more details).
Those protocols are very desirable since most applications of
quantum information theory require pure entangled states. Despite
its importance, purification protocols \cite{purereview} have been
only devised for stabilizer states \cite{2colgraphpur,3colgraphpur}, the W state
\cite{purw}, and LMESs \cite{CadV12}. In order to purify to arbitrary states, the spatially separated parties share initially the resource state, $\ket{\Phi_n}$ (e.g. the state in Eq. (\ref{res3}) for three--partite entanglement purification), where the controlling qubits are kept by one party. Then, the resource state, which is a stabilizer state, is purified. After that, the controlling qubits are measured to obtain the desired state shared between the spatially separated parties. In Appendix B we show that
this purification process outperforms previously known multipartite
purification protocols in the sense that the tolerable noise is much
higher.

We would like to thank W. D\"ur for fruitful discussions. The research was funded by the Austrian Science Fund (FWF): Y535-N16.

\section{APPENDIX A: Sending Classical Information}

The REP protocol does not only allow us to transmit quantum information but also classical information. In order to send classical information one restricts to REP of locally maximally entanglable states (LMESs) \cite{AppKrKr08}. Let us first review some of the basic properties of LMESs. Any $n$-qubit LMES can be written up to LUs in the form:
\bea\ket\Psi\label{LME} \equiv U \ket +^n=U_{1,\ldots,n} \prod U_{i_{k_1},\ldots,i_{k_n}}\,\ldots\,\prod U_i \ket +^{\otimes n},
\eea where $ U_{i_{k_1},\ldots,i_{k_n}}=e^{i \alpha_{i_{k_1},\ldots,i_{k_n}}\sigma_{3}^{(i_{k_1})}\otimes\ldots\otimes\sigma_3^{(i_{k_n})}}$.
Note that equivalently a LMES (up to LUs) can be written as $2^{-\frac{n}{2}}\sum_{i_1,\ldots,i_n=0}^1 e^{i\beta_{i_1\ldots i_n}}\ket{i_1\ldots i_n}$, with $\beta_{i_1\ldots i_n}\in\R$. In the following we will call a state with $\beta_{i_1\ldots i_n}\in\{0,\pi\}$ $\pi$--LMES.
It has been shown that for any $n$-qubit LMES, $\ket{\Psi}$ [Eq. (\ref{LME})], one can construct a generalized stabilizer, generated by  $n$ operators
$S_k, k\in\{1,\ldots, n\}$ \cite{AppKrKr08}. These hermitian and unitary operators are of the form $ S_k = U \sigma_1^{(k)} U^\dagger$ and their common eigenbasis is given by $\{\sigma_3^\mathbf{i}\,\ket{\Psi}\}_{\mathbf{i}\in \{0,1\}^n}$, where $\sigma_3^0 = \one$ and $\sigma_3^1=\sigma_3$, i.e. \bea \label{stab}S_k \sigma_3^\mathbf{i}\,\ket{\Psi} = (-1)^{i_k}\sigma_3^\mathbf{i}\,\ket{\Psi}\,\,\,\,\,\forall k.\eea
A LMES is called $k$-colourable if one can assign to every qubit one out of $k$ colours in such a way that no qubits interacting via a phase gate have the same colour \cite{AppCadV12}.\\

LMESs can be easily implemented via deteterministic gate implementation, since $U$ [Eq. (\ref{LME})] corresponds to a product of phase gates. As has been shown in the main text, the correction operators that can occur due to the implementation procedure are always of the form $\sigma_3^\mathbf{i}$ for $\mathbf{i}\in \{0,1\}^n$. In particular in the deterministic gate implementation of a general LMES of the form of Eq. (\ref{LME}) all possible Pauli corrections occur with the same probablity. This is easy to see since in deterministic gate implementation of a single--qubit unitary $U_l$ the correction $\one\otimes\ldots\otimes \sigma_3^{(l)}\otimes\ldots\one$ occurs with probability $\frac{1}{2}$, which ensures that in total the bitstring $\mathbf{i}$ corresponding to the correction operation $\sigma_3^\mathbf{i}$ will be random. \\
In order to transmit classical information, $A$ remotely prepares a LMES, $\ket{\Psi}$ (known to $B$), for $B$. So $B$ obtains one of the states $\sigma_3^\mathbf{i}\,\ket{\Psi}$. Note that the bitstring $\mathbf{i}$ only depends on $A$'s measurement results. Therefore, it is completely random and known to $A$.
We show next that, once $B$ receives a copy of the state $\sigma_3^\mathbf{i}\,\ket{\Psi}$, for some $\mathbf{i}\in \{0,1\}^n$, he can perform local measurements to obtain partial information on the bitstring.
The classical information that is sent corresponds to the part of the bitstring that $B$ can determine via local measurements.
In order to show that $B$ can get partial information on the bitstring, we make use of the following lemma.
\begin{lemma}\label{lemma1}

Let $\ket{\Psi_{\bf 0}}$ be a $\pi$--LMES, so that $\{\ket{\Psi_{\bf i}}=\sigma_3^{\bf i}\ket{\Psi_{\bf 0}}\}_{\mathbf{i}\in\{0,1\}^n}$ is the common eigenbasis of a generalized stabilizer generated by $\{S_j=\sigma_1^{(j)}\otimes U_j\}_{j=1}^n$. Then for any $\mathbf{k} \in \{0,1\}^{n-1}$ (denoting the computational basis of all qubits but qubit $j$) and $\ket{l_x^j}\in \{\ket{+},\ket{-}\}$ with $\bra{l_x^j {\bf k}}\ket{\Psi_{\bf i}}\neq 0$ it follows that $\bra{\Psi_{\bf i}}S_j\ket{\Psi_{\bf i}}=\bra{l_x^j {\bf k}}S_j\ket{l_x^j {\bf k}}$.
\end{lemma}

\begin{proof}\nonumber
Let $\ket{\Psi}$ be an element of the common eigenbasis. We write $\ket{\Psi}=\sum_{{\bf k}\in\{0,1\}^{n}} \ket{{\bf k}} \bra{{\bf k}}\Psi\rangle\equiv 2^{n-1/2} \sum_{{\bf k}\in\{0,1\}^{n-1}} \ket{{\bf k}} \ket{\phi_{\bf k}}$, with $\ket{\phi_{\bf k}}$ normalized. Due to the fact that $\ket{\Psi_{\bf 0}}$ is a $\pi$--LMES, it can be easily seen that $\ket{\phi_{\bf k}}\in \{\ket{+},\ket{-}\}\,\forall \mathbf{k}$. Then $\bra{\Psi}S_j\ket{\Psi}=
2^{-(n-1)} \sum_{{\bf k}\in\{0,1\}^{n}} \bra{\phi_{\bf k}}\sigma_1^{(j)} \ket{\phi_{\bf k}} \bra{{\bf k}} U^j \ket{{\bf k}}$, where we used the fact that $U_j$ is diagonal. Using that the expectation value of $S_j$ is $\pm 1$ and the fact that $|\bra{{\bf k}} U^j \ket{{\bf k}}|=1$ the equation above can only be fulfilled if $\bra{\phi_{\bf k}} \sigma_1^{(j)} \ket{\phi_{\bf k}} \bra{{\bf k}} U^j \ket{{\bf k}}=\langle S_j \rangle$ for any ${\bf k}$. Since furthermore $\bra{\phi_{\bf k}}\sigma_3\bra{\mathbf{k}}\ket{\Psi_{\bf i}}=0\,\forall \mathbf{k}$ this proves the statement.
\end{proof}
This Lemma allows us to identify the procedure for $B$ to learn part of the bitstring.
To determine $\langle S_j\rangle$ measure $\sigma_1$ on particle $j$ and $\sigma_3$ on the particles $N_j$ that are interacting with $j$. Denote by $M_j$ all particles that are neither $j$ nor in $N_j$. The state after the measurement corresponds to $\ket{l_x^j}_j\otimes\ket{\bf l}_{N_j}\otimes\ket{\phi^{l_x,\mathbf{l}}}_{M_j}$ for $\ket{\phi^{l_x,\bf{l}}}$ some state describing the remaining qubits and $\mathbf{l} \in \{0,1\}^{|N_j|}$. Note that we can choose, without loss of generality, $\ket{l_x^j{\mathbf{l}} 1\ldots 1}$ to compute the expectation value according to Lemma \ref{lemma1}, since $\bra{l_x^j {\mathbf {l}}1\ldots1}\ket{\Psi_{\bf i}}\neq 0$ and $S_j$ only acts on particle $j$ and the qubits interacting with it. Using Lemma \ref{lemma1} the expectation value of $\sigma_1^{(j)}\otimes U_j$ is given by $m_j m_{\mathbf{l}}$, where $m_j$ is the measurement outcome of party $j$ and $m_{\mathbf{l}}$ corresponds to $\bra{{\bf l}}U_j\ket{{\bf l}}$. This procedure requires that $B$ knows the generators of the generalized stabilizer of the LMES, $\ket{\Psi}$, that was prepared in deterministic gate implementation, i.e. he has to know $\ket\Psi$. Note that as long as the qubits $j_i$ are not interacting with each other this procudure allows to determine all $\langle S_{j_i}\rangle$ on a single copy of the state.

Note further that the measurements are local. \\
A $\pi$--LMES which is $k$-colourable and only involves $k$-qubit interactions is called $k$-regular LMES \cite{AppCadV12}. It is easy to see that they fulfill the above stated condition and, therefore, the expectation values corresponding to the same colour can be evaluated on a single copy. Hence, $k$-regular LMESs are useful in this context. This kind of classical information might be useful in the context of secure multipartite communication. \\
The resource state for the remote preparation of a LMES [Eq. (\ref{LME})] is maximally entangled in the splitting $A$ versus $B$, i.e. the reduced density matrix of $B$ is totally mixed and the entanglement between $A$ and $B$ is 1 ebit per qubit. Despite that fact, we have seen that $B$ can learn part of the bitstring via local measurements knowing that $A$ has performed her measurements beforehand.\\
As we have seen via choosing between different set-ups the REP protocol allows for the transmission of either quantum or classical information without actually sending it.

\section{APPENDIX B: Robust entanglement purification to arbitrary multipartite states}

While most applications in quantum information theory require
entangled pure states, actual implementations unavoidably introduce
noise to the target states. For this reason, purification protocols
to transform locally several copies of a mixed noisy state into a more
pure state have been thoroughly studied \cite{Apppurereview}. However,
with the sole exceptions of the $W$--state \cite{Apppurw} and the LMESs \cite{AppCadV12}, purification
protocols have been only devised for stabilizer states, which are states that are LU--equivalent to graph states.
Let us briefly recall the definition of graph states. Graph states correspond to $\pi$--LMESs, where each phase gate is acting only on 2 qubits. They are often associated with a mathematical graph consisting of a set of vertices, $V$, representing the qubits, and a set of edges, $E$, representing the two--qubit phase gates \cite{Appgraph}. Graph states constitute a notable class, for which
efficient purification procedures have been given
\cite{App2colgraphpur,App3colgraphpur}.
However, as is easy to understand,
not every mixed state can be purified locally and purification
protocols can only succeed if the noise lies below a certain
threshold.

Since our REP scheme relies on the parties sharing a stabilizer state and
then making local measurements on it to obtain an arbitrary
entangled state, this readily provides a purification protocol to
\textit{all} entangled states.
Furthermore, it turns out that the graph states corresponding to our resource states fall under
a class of graph states for which the existing purification
protocols are particularly robust against noise. Thus, we will show
that in certain scenarios using our scheme would allow to purify to
any given target state with significantly larger noise thresholds
than the purification protocols particularly devised for other
states such as the $3$--qubit $W$--state \cite{Apppurw} or LMESs \cite{AppCadV12}.

A usual scenario when studying purification schemes is the
following: A provider can prepare any desired entangled state, which
he then sends to the different distant parties that need to use it.
As each qubit is sent, it undergoes a local quantum channel which is
considered to be the main source of noise. We will consider here
that every qubit $i$ that is sent is subjected
to local depolarizing noise,
\begin{equation}\label{noise}
\mathcal{E}_p(\rho_i)=p\rho_i+\frac{1-p}4(\rho_i+\sigma_1\rho_i\sigma_1+\sigma_2\rho_i\sigma_2+\sigma_3\rho_i\sigma_3).
\end{equation}
We assume that the noise level $p$ is the same for all qubits that
need to be transmitted. Once the distant parties receive the noisy
state they can use LOCC to purify to their target state. To see more
clearly how our protocol would work, let us restrict ourselves to
the case of 3--qubit states for simplicity. Instead of sending the target state, A prepares the 8--qubit graph state corresponding to the resource stabilizer state given in Eq. (4) in the main text, which is needed to prepare
any pure 3--qubit entanglement and sends the corresponding 3 qubits
to the parties that constitute $B$ \footnote{One could also consider the case where one of the parties constituting $B$ holds the controlling qubits. In this case only two qubits need to be transmitted.}. These qubits undergo then some
noise, so $A$ and $B$ implement the corresponding graph-state
purification protocol to recover the original resource state
(4). After that, $A$ carries out the required measurements
and transmits the outcomes to the parties in $B$, who will then have
a more pure copy of the target state of choice. Therefore, one just needs
to study how robust is our graph state with the already devised
purification procedures for this kind of states.

It turns out that one of the most relevant properties of graph
states regarding purification is their colorability. A graph is
2--colorable if we can label each vertex using just two colors (say
$\mathcal{A}$ and $\mathcal{B}$) with the rule that no adjacent
vertices get the same color. $N$--colorable graphs are defined
analogously. Purification protocols for 2--colorable graph states
are simpler and usually more robust against noise (this is because
the qubits with different colors are purified independently and as
one purifies certain color, noise is introduced in the qubits of a
different color). As one can easily see from Fig.\ref{resource3}, our 8--qubit
resource graph state is 3--colorable.
\begin{figure}[h!]
\begin{center}
  \includegraphics[scale=0.20]{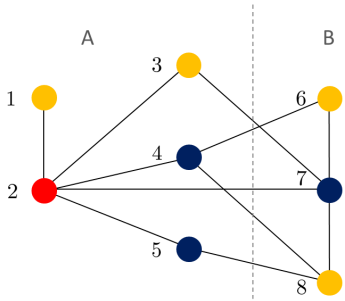}
  \caption{This graphic shows the graph state that is LU--equivalent to the 8-qubit resource state for REP in the 3-qubit case. }
\label{resource3}
\end{center}
\end{figure}

However, one may wonder
whether there exists a 2--colorable LU equivalent state to which the
parties can purify and after that apply local unitaries to obtain
the relevant graph state. All the different classes of graph states
up to 8 qubits have been studied in
\cite{Appgraphclasshein,Appgraphclassadan}. Using the invariants given in
\cite{Appgraphinvariantsadan}, it is lengthy but straightforward to
check that our graph belongs to the class number 98 in
\cite{Appgraphclassadan}. According to this reference, there is no
2--colorable graph in this class. Hence, one must use a 3--color
purification procedure in our case. However, as depicted in Fig.\ref{resource3},
one can color the qubits in such a way that the third color
$\mathcal{C}$ only appears once and not in the qubits that are sent
to $B$. It turns out that if the noise only affects the qubits that
are sent to $B$, the qubit with color $\mathcal{C}$ does not get
affected by the noise and, thus, a 2--colorable purification
procedure purifying qubits $\mathcal{A}$ and $\mathcal{B}$ is enough
to achieve the task. Implementing the procedure given in
\cite{App2colgraphpur}, one finds that the state can be purified if the
noise affecting the 3 transmitted qubits is such that $p\geq0.39$.
To make the model more realistic, we can also consider that the
qubits that are kept by $A$ undergo some (although weaker) form of
noise. Consider then similarly local depolarizing noise for these
other qubits with noise level $q$. Now, if $q\neq1$ all three colors
are affected by the noise and one then needs to actually carry out a
3-color purification protocol \cite{App3colgraphpur}. Although the
procedure is still relatively robust (as we will see below when
comparing to other protocols), the 3--colorability raises the noise
threshold $p$ for the transmitted qubits considerably even under
very weak forms of noise $q$ for the 5 qubits kept by $A$. In
particular, for $q=0.99$ one finds that purification succeeds when
$p\geq0.50$ while for $q=0.97$ the threshold is $p\geq0.56$.

However, interestingly, to purify locally to an arbitrary 3--qubit
state one does not need to be able to remotely prepare any 3--qubit
state. It is enough to be able to prepare any state in the $3$--qubit maximally entangled set as given in the
main text, and then transform it by LOCC to the desired state.
Therefore, it is enough to have the ability to purify to the
6--qubit resource graph state that allows for remote preparation of
the maximally entangled set.

\begin{figure}[h!]
\begin{center}
  \includegraphics[scale=0.25]{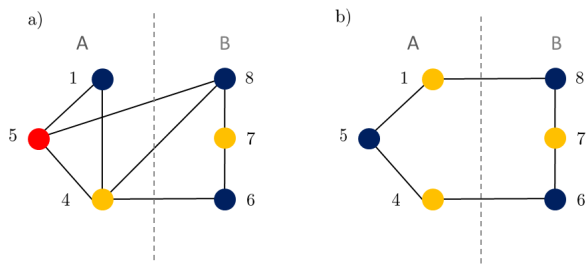}
  \caption{ $a)$ This graphic shows the graph state corresponding to the resource state for REP of the maximally entangled set. $b)$ Via local complementation one obtains the representative of class 18.}
\label{resourceMES}
\end{center}
\end{figure}

This graph state is also 3--colorable, but using again
the invariants of \cite{Appgraphinvariantsadan} one can find that this
state belongs to class 18, which includes 2-colorable graphs
\cite{Appgraphclasshein}. In fact, using the local complementation rule
\cite{Appgraphclasshein} it is easy to find the local unitaries that
transform our graph to the one shown in Fig.\ref{resourceMES} $b)$. One, therefore,
just needs to apply the 2--color purification protocol
\cite{App2colgraphpur}. This allows for much more robust purification
when all qubits are subjected to noise. In this case, the
corresponding thresholds are $p\geq0.44$ when $q=1$, $p\geq0.44$
when $q=0.99$ and $p\geq0.45$ when $q=0.97$.

As we have already stressed, not only the ability to purify to any
state is remarkable, but also the robustness of the scheme. To test
this, consider the purification protocols for LMESs \cite{AppCadV12}
and $W$ states \cite{Apppurw}, which are the few procedures known in
the literature besides graph states. For LMESs, we consider the
3--qubit regular LMES given by Eq.\ (\ref{LME}) when $U$ is the
3--qubit $\pi$-phase gate, e.g. $U=\one-2|111\rangle\langle111|$. If the
provider prepares this state and then sends it through the noisy
channel, the corresponding purification procedure works if the local
noise for the three qubits is such that $p\geq0.81$ \cite{AppCadV12},
which is noticeably worse than the threshold provided by the
procedure we give here for both cases, using the 8--qubit and the 6--qubit
resource graph states. The same happens in the case of the $W$, for
which the best-known local depolarizing noise threshold is
$p\geq0.69$
\cite{Apppurw}.
Moreover, the relative robustness of our scheme compared to these
protocols is of fundamental character as the latter cannot be
refined to achieve similar noise tolerances. To see this, consider
the case of the $W$ state. If one sends this state through the
quantum channel given by (\ref{noise}), no purification is possible
when the local noise level is such that $p\leq0.58$. This is
because, the output of this channel applied to the $W$ state is PPT
and, hence, non-distillable. Remarkably, our scheme would allow to
accomplish the task under channels with such (and even considerably
bigger) amounts of noise for the transmitted qubits.

So far we have considered the situation where $A$ wants to provide the three parties in $B$, i.e. $B_1,B_2,B_3$ with an arbitrary entangled state. However, regarding entanglement purification, one could also consider the situation where the final goal is that $A$ shares the target state with $B_1$ and $B_2$. In this case, only $2$ qubits need to be send. In case the $8$--qubit state is used, the thresholds change to $ p=0.46$ (compared to $p=0.50$ from above) if the non-transmitted qubits undergo noise with $q=0.99$ and to  $p=0.52$ for $q=0.97$ (compared to $p=0.56$). The change is even much more notable if the 6--qubit state is utilized. In this case, the threshold changes from $p=0.44$ to $p=0.34$ (assuming no noise on the qubits which are not transmitted).

Finally, notice that, of course, a similar purification scheme can
be established using teleportation. For that, $A$ would send one
share of a bipartite maximally entangled state to each party in $B$.
After the noisy transmission these pairs can be purified back to
maximally entangled states and then $A$ would simply teleport the
target state. The purification to bipartite maximally entangled
states succeeds whenever $p>1/3$ in our example (assuming no noise
in the qubits held by $A$), which offers a better noise tolerance
than our REP-based scheme. However, notice that the teleportation
strategy demands that the parties in $A$ remain together (otherwise
they cannot prepare and teleport the target state). On the other
hand, the REP strategy only requires local operations from the
parties in $A$ and, hence, once the resource graph state has been
distributed they can be spatially separated and still prepare any
target state. Despite this fact, the REP leads to almost the same threshold as the teleportation scheme. Since this case would be comparable to our protocol if $A$ keeps one of the three qubits for which the threshold is $p>0.34$, as stated above.

\end{document}